\documentclass[lettersize, journal]{IEEEtran}
\usepackage{amsmath,amsfonts}
\usepackage{algorithm2e}
\usepackage{array}
\usepackage[caption=false,font=normalsize,labelfont=sf,textfont=sf]{subfig}
\usepackage{textcomp}
\usepackage{stfloats}
\usepackage{url}
\usepackage{verbatim}
\usepackage{graphicx}
\usepackage{cite}
\usepackage{subfig} 
\usepackage{makecell}
\usepackage[top=0.69in, bottom=1.02in, left=0.63in, right=0.63in]{geometry}
\usepackage{multirow} 

\usepackage{amsthm}

\usepackage{graphicx} 
\usepackage{amsmath}
\usepackage{amsmath}
\usepackage{amsmath, amssymb}
\usepackage{xcolor}	
\usepackage{amsmath} 

\usepackage{algpseudocode}
\usepackage{graphics}
\usepackage{amsmath, amsthm}

\theoremstyle{plain}
\newtheorem{theorem}{\textit{Theorem}}
\newtheorem{lemma}{Lemma}

\begin{document}

	\title{Copula-Based Analysis of Fluid Antenna-Assisted Over-the-Air Computation}
	
	\author{IEEE Publication Technology,~\IEEEmembership{Staff,~IEEE,}
	}

	\author{Saeid Pakravan, Mohsen Ahmadzadeh, Wessam Ajib, Ghosheh Abed Hodtani, Ming Zeng

\thanks{S. Pakravan and M. Zeng are with the Department of Electric and Computer Engineering, Laval University, Quebec City, QC, Canada. email: saeid.pakravan.1@ulaval.ca; ming.zeng@gel.ulaval.ca.}

	\thanks{M. Ahmadzadeh and G. Abed Hodtani are with the Department of Electric and Computer Engineering, Ferdowsi University, Mashhad, Iran. email: m.ahmadzadehbolghan@mail.um.ac.ir; hodtani@um.ac.ir.}

\thanks{W. Ajib is with the Department of Computer Science, University of Quebec in Montreal (UQAM), Montreal, QC, Canada. email: ajib.wessam@uqam.ca.}

	}

	\maketitle

	\begin{abstract}
		
This letter studies an uplink over-the-air computation (AirComp) framework in which multiple user equipments are equipped with fluid-antenna (FA) arrays and operate over spatially correlated fading channels. By explicitly modeling channel dependence using the Gumbel copula, closed-form analytical expressions are derived for the cumulative distribution function (CDF) of the mean-squared error (MSE) of the aggregated function. The proposed analysis provides a quantitative performance characterization of AirComp under spatial correlation and provides analytical insights into the role of FA-assisted transmission in correlated wireless environments. Numerical results validate the derived expressions and show that FA deployment can substantially reduce the MSE compared with conventional fixed-antenna systems, although the achievable gain decreases as spatial correlation becomes stronger.

\end{abstract}


\section{Introduction}

Over-the-air computation (AirComp) has emerged as an enabling paradigm for efficient wireless data aggregation by exploiting the signal superposition property of multiple-access channels~\cite{10092857}. By allowing multiple devices to transmit analog signals simultaneously and aggregating them directly in the air, AirComp eliminates the need for orthogonal channel access and explicit decoding of individual messages, thereby significantly reducing communication latency and improving spectral efficiency~\cite{10538293}. These characteristics make AirComp particularly attractive for real-time and large-scale data aggregation in Internet-of-Things (IoT) networks~\cite{9535447}. However, the performance of AirComp is inherently sensitive to wireless impairments such as channel fading, interference, and noise, which directly degrade the accuracy of function aggregation.

To address these challenges, extensive research has focused on propagation-aware AirComp designs that aim to minimize the mean-squared error (MSE) of the aggregated signal. Representative approaches include the use of reconfigurable intelligent surfaces (RIS)~\cite{10974462} and advanced beamforming or directional transmission techniques in multi-antenna systems~\cite{li2023integrated}, which exploit spatial diversity and controllable propagation environments to enhance aggregation reliability. These methods improve AirComp performance by shaping the wireless environment or concentrating signal energy along favorable propagation directions. While highly effective, such approaches typically rely on infrastructure support or coordinated transmission strategies, which may introduce additional system complexity or deployment constraints in highly dynamic scenarios.

Recently, fluid-antenna (FA) systems have been introduced as a promising alternative for enhancing spatial diversity in compact wireless devices~\cite{Pakravan20261, 10906511, 11106811, zhang2024fluid}. By enabling antenna elements to dynamically reposition within a confined region, FA systems provide additional spatial degrees of freedom without increasing hardware footprint. In contrast to infrastructure-assisted approaches such as RIS or beamforming, FA systems enable device-level spatial reconfigurability, allowing each user to locally adapt to channel variations through antenna position selection. This capability has motivated the integration of FA arrays into AirComp frameworks, demonstrating improved robustness and adaptability under time-varying channel conditions~\cite{saeidpwcnc,pakravan2024robust, 10729877, ahmadzadeh2025ai}. Compared with conventional fixed-position antennas (FPA), FA-assisted transmission offers greater flexibility in exploiting favorable channel realizations, which is particularly beneficial for analog aggregation schemes such as AirComp.

Despite these advances, the theoretical understanding of FA-assisted AirComp remains incomplete, especially with respect to the role of spatial correlation among FA ports. In practical implementations, FA ports are closely spaced within a compact region, leading to non-negligible spatial correlation that fundamentally affects channel statistics and aggregation performance~\cite{10623405, 10678877, huangfu2025performance}. Existing works often rely on simplified independence assumptions or numerical evaluations, leaving a gap in rigorous analytical characterization.

In this paper, we address this gap by developing a unified analytical framework for FA-assisted AirComp systems under spatially correlated fading channels. In contrast to existing works that treat FA systems, AirComp, or correlation modeling in isolation, the proposed approach jointly integrates these aspects into a single tractable framework. Using copula theory~\cite{10678877, 10319727}, we first derive analytical expressions for the joint cumulative distribution function (CDF) of FA channel gains under correlated Rayleigh fading, explicitly capturing the statistical dependence among FA ports. Based on this framework, compact closed-form expressions are obtained for the CDF of the aggregation MSE. While the present derivation is specialized to correlated Rayleigh fading, the copula-based framework can be extended to other marginal fading distributions. This enables a direct and fair comparison between FA-assisted and conventional FPA-based AirComp systems. Extensive Monte-Carlo (MC) simulations, driven by copula-based random generation, validate the analytical results and demonstrate the pronounced impact of spatial correlation on AirComp performance, highlighting the substantial MSE reduction achievable through FA deployment under different correlation regimes.
To better highlight the novelty of the proposed framework, Table~\ref{tab:comparison} summarizes the key differences between this work and existing studies.

\begin{table}[t]
\caption{Comparison with Existing Works}
\label{tab:comparison}
\centering
\footnotesize
\setlength{\tabcolsep}{4pt}
\renewcommand{\arraystretch}{1.15}

{
\begin{tabular}{c|c|c|c|c|c}
\hline
\textbf{Ref.} & \textbf{AirComp} & \textbf{FA} & \textbf{Corr.} & \textbf{Copula} & \textbf{Closed-Form MSE} \\
\hline

{[1]--[5]}    & \checkmark & \texttimes & \texttimes & \texttimes & \texttimes \\
{[6]--[8]}    & \texttimes & \checkmark & \texttimes & \texttimes & \texttimes \\
{[9]--[13]}  & \checkmark & \checkmark & \texttimes & \texttimes & \texttimes \\
{[14], [16]}  & \texttimes & \checkmark & \checkmark & \texttimes & \texttimes \\
{[15], [17]}  & \texttimes & \checkmark & \checkmark & \checkmark & \texttimes \\
\hline
\textbf{This Work} 
& \checkmark & \checkmark & \checkmark & \checkmark & \checkmark \\
\hline
\end{tabular}
}

\vspace{1ex}
\begin{minipage}{0.95\linewidth}
\footnotesize
{\textbf{Corr.}: Spatial correlation modeling.}
\end{minipage}
\end{table}

\section{System Model}

We consider an uplink AirComp system consisting of $K$ single-antenna user equipments (UEs), denoted by $\mathrm{UE}_k$, $k \in \mathcal{K} \triangleq \{1,2,\ldots,K\}$, communicating with an access point (AP) equipped with a single FPA. Each UE is equipped with a single FA, which
enables dynamic selection among multiple antenna ports located within a compact
spatial region.

Each $\mathrm{UE}_k$ aims to transmit a data vector $\mathbf{s}_k \in \mathbb{C}^{r}$, where $r$ denotes the signal dimension. The transmitted signals satisfy $\mathbb{E}[\mathbf{s}_k]=\mathbf{0}$, $\mathbb{E}[\mathbf{s}_k \mathbf{s}_k^{H}]=\mathbf{I}_r$, and $\mathbb{E}[\mathbf{s}_k \mathbf{s}_j^{H}]=\mathbf{0}$ for $k \neq j$. The vector $\mathbf{s}_k$ represents $r$ function values to be aggregated and is transmitted over $r$ orthogonal channel uses (e.g., time slots or subcarriers), which is consistent with practical single-RF-chain FA implementations.
The objective of AirComp is to compute the aggregated statistic of the transmitted data vectors at the AP~\cite{zhang2024fluid}, given by
\begin{equation}
	\mathbf{s} =\sum_{k=1}^{K} \mathbf{s}_k.
\end{equation}

Each FA consists of $N$ antenna ports uniformly distributed along a linear aperture of length $W\lambda$, where $\lambda$ denotes the carrier wavelength. The relative position of the $n$-th port is 
\begin{equation}
	d_n = \frac{n-1}{N-1} W\lambda, \qquad n = 1,2,\ldots,N.
\end{equation}

Let $h_{k,n}$ denote the channel coefficient between the $n$-th FA port of $\mathrm{UE}_k$ and the AP. Due to the compact spatial arrangement of the FA ports, the channel coefficients across different ports exhibit spatial correlation~\cite{10678877, 10319727}.\footnote{{The adopted model captures spatial correlation among FA ports but does not explicitly account for electromagnetic mutual coupling. This abstraction is used to preserve analytical tractability. Incorporating mutual coupling into the FA channel model, particularly for dense port configurations, is an important direction for future work.}}

At each transmission instance, each UE selects the FA port that maximizes its instantaneous
channel power gain~\cite{huangfu2025performance}.\footnote{{The model assumes full CSI across FA ports for optimal selection. In practice, this may incur non-negligible overhead, especially in fast-fading environments. This overhead can be reduced via partial port scanning, low-complexity selection, or learning-based prediction. Moreover, imperfect CSI may lead to suboptimal port selection and degraded effective channel gains, thereby increasing the aggregation MSE. Incorporating these practical aspects into the analysis is an important direction for future work.}} Accordingly, the effective channel gain of $\mathrm{UE}_k$ is defined as
\begin{equation}
\label{eq:g_k}
g_k=\max_{1\le n\le N} g_{k,n},
\end{equation}
where $g_{k,n}\triangleq|h_{k,n}|^2$ denotes the channel power gain associated with the
$n$-th FA port. Let $n_k^\star$ denote the index of the selected port of $\mathrm{UE}_k$,
i.e.,
\begin{equation}
	n_k^\star
	=
	\arg\max_{1 \le n \le N} |h_{k,n}|^2 .
\end{equation}
Let $h_k \triangleq h_{k,n_k^\star}$ denote the effective channel after FA port selection. In the sequel, all expressions are written in terms of $h_k$ for notational clarity.
Accordingly, the received signal at the AP is expressed as
\begin{equation}
	\mathbf{y} = \sum_{k=1}^{K} p_{k}\, h_{k}\, \mathbf{s}_{k} + \mathbf{z},
\end{equation}
where $\mathbf{z} \sim \mathcal{CN}(0, \sigma^2 \mathbf{I}_r)$ denotes additive white Gaussian noise (AWGN), and $p_k$ denotes the complex transmit scaling coefficient of $\mathrm{UE}_k$. Each UE satisfies the average power constraint~\cite{chen2023joint}
\begin{equation}
\label{power_constraint}
\frac{1}{r}\,\mathbb{E}\!\left[\|p_k \mathbf{s}_k\|^2\right]
= |p_k|^2 \;\le\; p_{\max}, \qquad \forall k \in \mathcal{K}.
\end{equation}

The AP estimates the aggregated signal using a normalization factor $\rho$ as
\begin{equation}
	\label{rt}
	\hat{\mathbf{s}}
	=\frac{1}{\sqrt{\rho}}\sum_{k=1}^{K}
	p_{k} h_{k} \mathbf{s}_{k}
	+\frac{1}{\sqrt{\rho}} \mathbf{z}.
\end{equation}

The aggregation error is quantified by the MSE, defined as
\begin{equation}
	\label{mse}
		\mathrm{MSE}
\triangleq
\frac{1}{r}\mathbb{E}\!\left[\left\| \mathbf{s} - \hat{\mathbf{s}} \right\|^2 \right]
		=
		\sum_{k=1}^{K}
		\left| 1 - \frac{1}{\sqrt{\rho}}\, p_{k} h_{k} \right|^{2}
		+ \frac{\sigma^{2}}{\rho}.
\end{equation}

Following~\cite{chen2023joint}, a zero-forcing transmit scaling is adopted to eliminate the signal misalignment across UEs, given by
\begin{equation}
	\label{pk}
	p_{k} = \sqrt{\rho}\frac{h_{k}^H}{|h_{k}|^2}.
\end{equation}
To satisfy the power constraint in \eqref{power_constraint}, the normalization factor must
satisfy
\begin{equation}
\label{ro}
	\rho \le  p_{\max} |h_{k}|^2, \quad \forall k.
\end{equation}

Substituting \eqref{pk} and \eqref{ro} into \eqref{mse}, the resulting aggregation MSE
admits the following closed-form expression~\cite{pakravan2024robust, ahmadzadeh2025ai}:
\begin{equation}
	\label{lll}
	\mathrm{MSE}
	= \frac{\sigma^2}{p_{\max}}
	\max_{k \in \mathcal{K}}
	\frac{1}{|h_{k}|^2}.
\end{equation}
{This expression shows that the aggregation MSE is governed by the user with the weakest effective channel gain, i.e., the minimum $|h_k|^2$. This bottleneck effect arises from channel inversion–based AirComp, where all users align their signals under a common scaling constrained by the worst channel. Consequently, the overall accuracy is limited by the weakest link. This highlights the importance of improving the minimum effective channel gain. In this context, FA-enabled users mitigate this effect by exploiting position-domain diversity to select favorable ports, thereby enhancing their effective channel gains and reducing the MSE.}

{In the following, we build on this expression to statistically characterize the aggregation MSE under spatially correlated fading.
}

\section{Theoretical Analysis}

{
This section provides a statistical characterization of the aggregation MSE for the proposed FA-assisted AirComp system. The analysis starts from the closed-form expression in \eqref{lll} and examines the distribution of the effective channel gains under spatially correlated fading. By incorporating the dependence among FA ports via copula theory, the CDF of the aggregation MSE is derived in a tractable form.}

For analytical convenience, we define the auxiliary random variable
\begin{equation}
	X_k \triangleq \frac{\sigma^2}{p_{\max} g_k}.
\end{equation}
Using \eqref{lll}, the CDF of the MSE can be expressed as
\begin{equation}
\small
	\label{eq:MSE_CDF}
	F_{\mathrm{MSE}}(\tau)
	= \Pr(\mathrm{MSE}<\tau)
	= \Pr\!\left(\max_{k\in\mathcal{K}} X_k < \tau\right),
\end{equation}
where $\tau > 0$ represents a prescribed MSE threshold corresponding to the
desired aggregation accuracy. Assuming independent fading across UEs, \eqref{eq:MSE_CDF} simplifies to\footnote{{Independent fading across UEs is assumed for tractability. Inter-user correlation in dense scenarios may reduce diversity and degrade performance due to the weakest-user effect. Its incorporation is left for future work.}} 
\begin{equation}
	F_{\mathrm{MSE}}(\tau)
	= \prod_{k=1}^{K} \Pr(X_k<\tau)
	= \big(F_{X_k}(\tau)\big)^K .
\end{equation}
The CDF of $X_k$ is obtained as
\begin{equation}
	\label{eq:F_Xk}
	F_{X_k}(x)
	= \Pr\!\left(g_k>\frac{\sigma^2}{p_{\max}x}\right)
	= 1 - F_{g_k}\!\left(\frac{\sigma^2}{p_{\max}x}\right),
\end{equation}
where $F_{g_k}(\cdot)$ denotes the CDF of the effective FA channel power gain.

Under Rayleigh fading, the channel power gain at each FA port follows an exponential distribution with unit mean, i.e., $g_{k,n}\sim\mathrm{Exp}(1)$, yielding the marginal CDF
$
F_{g_{k,n}}(x)=1-e^{-x}, \; x\ge0.
$
Due to the compact aperture of the FA, the gains $\{g_{k,n}\}_{n=1}^{N}$ are generally spatially correlated. To accurately capture this dependence structure while preserving the marginal
distributions, we adopt a copula-based statistical modeling framework.

\begin{theorem}[Sklar’s Theorem~\cite{nelsen2006introduction}]
Let $\mathbf{X}=(X_1,\ldots,X_J)$ be a random vector with joint CDF $F_{\mathbf{X}}$ and continuous marginal CDFs $F_{X_j}$. Then, there exists a unique copula function $C:[0,1]^J\rightarrow[0,1]$ such that
\begin{equation}
	\label{eq:sklar}
	F_{\mathbf{X}}(x_1,\ldots,x_J)
	=
	C\!\big(F_{X_1}(x_1),\ldots,F_{X_J}(x_J)\big).
\end{equation}
\end{theorem}

Among existing copula families, the Gumbel copula is particularly suitable for the present FA-assisted AirComp setting because FA port selection is governed by the maximum channel gain across ports. Hence, the resulting performance is primarily influenced by the upper-tail dependence of the joint channel distribution. 
{In particular, the Gumbel copula captures upper-tail dependence, enabling accurate modeling of extreme (high-gain) channel realizations that dominate the selection process. In contrast, Gaussian copulas do not exhibit tail dependence and thus cannot effectively characterize such extreme-value behavior, while Clayton copulas emphasize lower-tail dependence, which is less relevant for the considered maximum-gain selection mechanism~\cite{10678877, silva2008copula}.
Therefore, a copula with explicit upper-tail dependence is more appropriate for modeling the dependence structure in FA-assisted AirComp.}

Specifically, we employ the Archimedean Gumbel copula of dimension $d$, defined as~\cite{nelsen2006introduction}
\begin{equation}
\small
	\label{eq:gumbel_copula}
	C_{\mathrm{G}}(u_1,\ldots,u_d)
	=
	\exp\!\left(
	-
	\Bigg(
	\sum_{i=1}^{d}(-\ln u_i)^{\theta}
	\Bigg)^{1/\theta}
	\right),
\end{equation}
where $\theta \ge 1$ is the dependence parameter controlling the strength of
spatial correlation, and
$u_i = F_{X_i}(x_i) \in [0,1]$ denotes the probability–integral–transformed marginal
of the $i$-th FA port channel gain. 
{It is important to note that the copula dependence parameter is not a direct physical quantity, but rather a statistical measure of dependence among FA-port channel gains. To relate $\theta$ to physical channel characteristics, it can be calibrated through Kendall’s $\tau$, which provides a well-established link between statistical dependence and correlation observed in spatial channel models or measurements. For the Gumbel copula, Kendall’s $\tau$ is related to the dependence parameter through $\tau = 1 - \frac{1}{\theta}$, establishing a direct mapping between correlation levels and $\theta$. In practice, this mapping can be implemented by estimating the empirical Kendall’s $\tau$ from simulated or measured FA-port channel samples, and then obtaining the corresponding value of $\theta$. This procedure provides a practical bridge between physically observed spatial correlation and the copula-based dependence model~\cite{nelsen2006introduction}. In particular, the level of dependence is influenced by the spatial configuration of FA ports, where larger apertures or wider port spacing typically lead to weaker correlation, while closely spaced ports result in stronger dependence and 
larger $\theta$.
}

Using the Gumbel copula model, the distribution of the effective FA channel gain
can be characterized as follows.

\begin{lemma}
	\label{lemma:gumbel_cdf}
	Let $\{g_{k,n}\}_{n=1}^{N}$ be exponentially distributed random variables with unit mean, whose dependence is modeled by a Gumbel copula with parameter $\theta\ge1$. Then, the CDF of
	$
	g_k=\max\{g_{k,1},\ldots,g_{k,N}\}
	$
	is given by
	\begin{equation}
    \small
		\label{eq:gk_cdf}
		F_{g_k}(x) =
\exp\!\left(
- \left(
\sum_{n=1}^{N}
\bigl[-\ln(1-e^{-x})\bigr]^{\theta}
\right)^{1/\theta}
\right).
	\end{equation}
\end{lemma}

\begin{proof}
	The result follows by applying Sklar’s theorem to the joint distribution of $\{g_{k,n}\}_{n=1}^{N}$ and substituting the marginal CDF $F_{g_{k,n}}(x)=1-e^{-x}$ into the Gumbel copula in \eqref{eq:gumbel_copula}.
\end{proof}
{The structure of \eqref{eq:gk_cdf} shows that spatial correlation is captured through an exponent-like scaling across the $N$ FA ports, effectively behaving as $N^{1/\theta}$. While $\theta=1$ yields full diversity with linear scaling in $N$, increasing $\theta$ reduces the effective scaling and thus the achievable diversity gain.}


Substituting \eqref{eq:gk_cdf} into \eqref{eq:F_Xk} and \eqref{eq:MSE_CDF}, the CDF of the aggregation MSE is obtained in closed form as
\begin{equation}
\small
	\label{eq:MSE_CDF_final}
	F_{\mathrm{MSE}}(\tau)
	=
	\Bigg[
	1 -
	\exp\!\Bigg(
	-
	\Bigg(
	\sum_{n=1}^{N}
	\big(-\ln(1-e^{-\sigma^2/(p_{\max}\tau)})\big)^{\theta}
	\Bigg)^{1/\theta}
	\Bigg)
	\Bigg]^K.
\end{equation}
The dependence parameter $\theta$ plays a key role in determining
AirComp performance. When $\theta = 1$, the Gumbel copula reduces to the
independence case, corresponding to uncorrelated FA ports. In this regime, the system achieves the maximum spatial diversity gain and the minimum aggregation
MSE.

As $\theta$ increases, positive dependence among FA ports becomes stronger,
which reduces the effective diversity gain obtained through port selection and
leads to degraded MSE performance. In the limiting case $\theta \to \infty$, the
Gumbel copula converges to the Fr\'echet--Hoeffding upper bound, representing
perfect positive dependence. Under this condition, all FA ports experience
identical fading realizations, no spatial diversity is available, and the
FA-assisted AirComp system reduces to the performance of a conventional
FPA system. This scenario corresponds to the worst-case
aggregation MSE under severe spatial correlation.

\section{Numerical Results}

In this section, we evaluate the performance of the proposed FA-assisted AirComp framework through MC simulations and investigate the impact of spatial channel correlation on the aggregation MSE. Unless otherwise stated, all results are obtained by averaging over $10^4$ independent MC realizations.

We consider an uplink AirComp scenario with $K=10$ single-antenna UEs, each equipped with a FA system with $N=10$ selectable ports. These parameters are consistent with recent FA-assisted AirComp studies~\cite{11106811, ahmadzadeh2025ai, 10678877}. The proposed scheme is benchmarked against a conventional FPA system, which corresponds to the fully correlated case. {For consistency with the normalized analytical model, the noise power is set to $\sigma^2=1$, and the transmit power budget is expressed in the corresponding normalized linear scale. Unless otherwise stated, we set $p_{\max}=10$, corresponding to a 10-dB transmit-power-to-noise ratio. The spatial correlation among FA ports is controlled by the copula parameter $\theta$, through which different correlation regimes are examined in the numerical results.}

Spatial correlation among FA ports is modeled using an Archimedean copula-based channel generation framework, which
enables flexible characterization of statistical dependence structures beyond conventional Gaussian correlation models
\cite{nelsen2006introduction}. The detailed random variate generation procedure based on standard Archimedean copula sampling is provided in Appendix~A.

By tuning the copula dependence parameter $\theta$, different spatial correlation regimes can be systematically examined. Specifically, $\theta = 1$ corresponds to independent fading across antenna ports and yields the maximum achievable spatial diversity. Increasing values of $\theta$ indicate progressively stronger positive spatial correlation. In the limiting case $\theta \rightarrow \infty$, the copula converges to the Fr\'echet--Hoeffding upper bound, modeling the conventional FPA system in which all antenna ports experience identical fading
and no diversity gain is available.

Fig.~\ref{fig:copula_a} compares the analytical CDF derived in \eqref{eq:MSE_CDF_final} with the empirical CDF obtained from MC simulations for different target thresholds $\tau$. The close agreement between theory and simulations confirms the accuracy of the proposed copula-based analytical framework. Moreover, the aggregation MSE increases monotonically with $\theta$, indicating that stronger spatial correlation reduces the effective diversity gain provided by FA port selection. Nevertheless, the FA-assisted AirComp scheme consistently outperforms the conventional FPA baseline across all considered correlation levels.

Fig.~\ref{fig:copula_b} illustrates the impact of the number of FA ports $N$ on AirComp performance for a fixed threshold $\tau=0.3$ and $K=10$ UEs. When $N=1$, the system reduces to the conventional FPA case. As $N$ increases, the aggregation MSE decreases due to the increased spatial diversity provided by the FA. However, the performance improvement gradually saturates as $N$ becomes large. This behavior can be attributed to the spatial correlation among FA ports within the confined aperture. Specifically, although increasing $N$ enlarges the candidate set for port selection, the corresponding channel gains become increasingly correlated, which limits the effective diversity gain. As a result, the improvement in the effective channel gain $g_k$ becomes marginal beyond a certain point, leading to a saturation in the MSE performance. Furthermore, the independent fading scenario achieves the best performance, while stronger spatial correlation progressively limits the achievable gains.

Fig.~\ref{fig:copula_c} depicts the effect of the number of UEs $K$ on the FA-assisted AirComp performance for $\tau=0.3$ and $N=10$. As $K$ increases, the aggregation MSE degrades since the overall performance is dominated by the weakest effective channel among the users. The conventional FPA system consistently exhibits the poorest performance, underscoring the advantage of FA deployment. Furthermore, the impact of spatial correlation becomes more pronounced as the number of users grows, highlighting the importance of correlation-aware FA design in large-scale AirComp systems.

\begin{figure*}[!t]
    \centering
    \subfloat[]{
        \includegraphics[width=0.31\textwidth,height=0.19\textwidth]{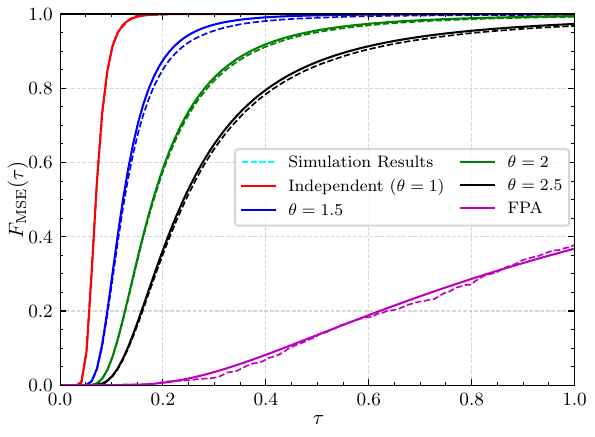}
        \label{fig:copula_a}
    }\hfill
    \subfloat[]{
        \includegraphics[width=0.31\textwidth,height=0.19\textwidth]{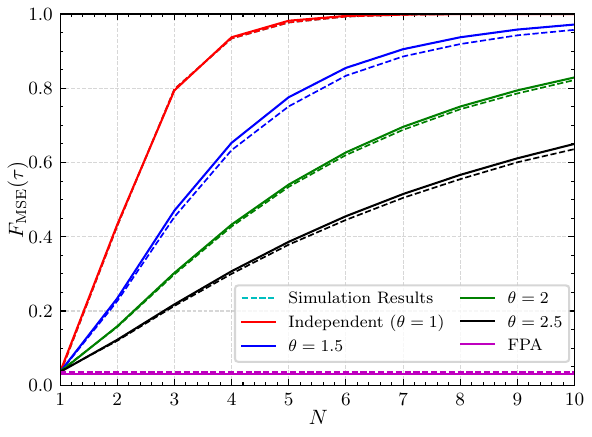}
        \label{fig:copula_b}
    }\hfill
    \subfloat[]{
        \includegraphics[width=0.31\textwidth,height=0.19\textwidth]{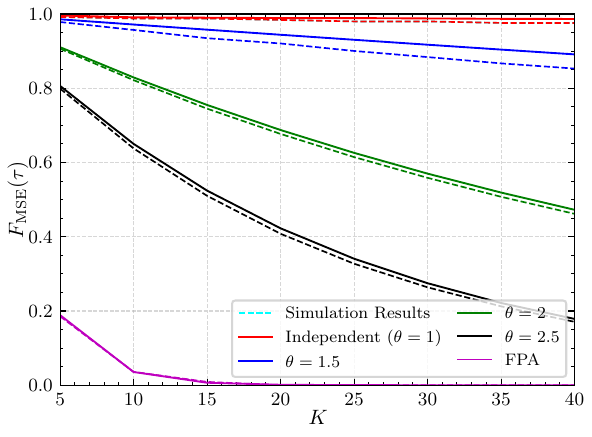}
        \label{fig:copula_c}
    }
    \caption{Statistical characterization of the aggregation MSE under copula-based spatial correlation: (a) CDF of the aggregation MSE versus the target threshold $\tau$; (b) CDF of the aggregation MSE versus the number of FA ports $N$; (c) CDF of the aggregation MSE versus the number of users $K$.}
    \label{fig:copula_results}
\end{figure*}

\section{Conclusion}

This letter developed a copula-based analytical framework for FA-assisted AirComp under spatially correlated fading. Using the Gumbel copula to capture the upper-tail dependence relevant to FA port selection, we derived a closed-form CDF for the aggregation MSE and validated it via simulations. The results showed that spatial dependence fundamentally limits the diversity gain offered by FA selection, while FA-assisted AirComp still achieves clear performance gains over conventional FPA systems. These findings provide useful design insights for correlation-aware AirComp in compact wireless devices. Future work may extend the proposed framework to imperfect CSI, inter-user channel dependence, and more general fading environments.

{
\appendices
\section{}\label{app:copula}

To generate spatially correlated FA-port channel gains, we adopt the frailty representation of the Archimedean Gumbel copula\cite{nelsen2006introduction}, which enables efficient sampling of dependent random variables with prescribed marginals and dependence structure. 

For the Gumbel copula with dependence parameter $\theta \geq 1$, the generator is given by
\begin{equation}
	\phi(t)=(-\ln t)^{\theta}, \quad t \in (0,1],
\end{equation}
and the corresponding $N$-dimensional copula is expressed as
\begin{equation}
\small
	C_G(u_1,\ldots,u_N)
	=
	\exp\left(
	-\left[
	\sum_{n=1}^{N}(-\ln u_n)^\theta
	\right]^{1/\theta}
	\right).
\end{equation}
Let $\alpha=1/\theta$. For each UE, generate an independent positive $\alpha$-stable random variable $V_k$ whose Laplace transform satisfies
\begin{equation}
	\mathbb{E}\left[e^{-qV_k}\right]=\exp(-q^{\alpha}),
	\qquad q \geq 0.
\end{equation}
Next, generate $N$ independent exponential random variables
$E_{k,n}\sim \mathrm{Exp}(1)$, $n=1,\ldots,N$. The correlated copula samples for UE $k$ are then obtained as
\begin{equation}
	u_{k,n}
	=
	\exp\left[
	-\left(\frac{E_{k,n}}{V_k}\right)^{\alpha}
	\right],
	\qquad n=1,\ldots,N.
\end{equation}

Finally, the correlated Rayleigh fading channel power gains are obtained through inverse transform sampling as
\begin{equation}
	g_{k,n}
	=
	-\ln(1-u_{k,n}),
	\qquad n=1,\ldots,N,
\end{equation}
which ensures that each $g_{k,n}$ follows an exponential distribution with unit mean while preserving the Gumbel dependence structure across FA ports.

Independent realizations of $V_k$ are generated for different UEs, ensuring independence across users, which is consistent with the system model assumptions.
}

	\bibliographystyle{IEEEtran}
	\bibliography{Main_Document}

@ARTICLE{10538293,
	author={Wang, Zhibin and others},
	journal={IEEE Internet Things J.}, 
	title={Over-the-Air Computation for \text{6G}: Foundations, Technologies, and Applications}, 
	month={May},
	year={2024},
	volume={11},
	number={14},
	pages={24634-24658}}

@ARTICLE{10092857,
  author={Şahin, Alphan and others},
  journal={IEEE Commun. Surveys Tuts.}, 
  title={A Survey on Over-the-Air Computation},
month={Apr.},
  year={2023},
  volume={25},
  number={3},
  pages={1877-1908}}

@ARTICLE{9535447,
  author={Zhu, Guangxu and others},
  journal={IEEE Wireless Commun.}, 
  title={Over-the-Air Computing for Wireless Data Aggregation in Massive {IoT}},
month={Sep.},
  year={2021},
  volume={28},
  number={4},
  pages={57-65}}

@ARTICLE{10974462,
  author={Xiao, Yue and others},
  journal={IEEE Trans. Wireless Commun.}, 
  title={{RIS}-Assisted Multi-Cell Over-the-Air Computation}, 
month={Apr.},
  year={2025},
  volume={24},
  number={9},
  pages={7437-7452}}

@article{Pakravan20261,
  author  = {Pakravan, Saeid and others},
  title   = {Fluid Antenna Systems under Channel Uncertainty and Hardware Impairments: Trends, Challenges, and Future Research Directions},
  journal = {arXiv preprint arXiv:2601.22989},
  year    = {2026},
  month   = {Jan.}
}

@ARTICLE{10623405,
  author={Ramírez-Espinosa, Pablo and others},
  journal={IEEE Trans. Wireless Commun.}, 
  title={A New Spatial Block-Correlation Model for Fluid Antenna Systems}, 
month={Aug.},
  year={2024},
  volume={23},
  number={11},
  pages={15829-15843}}

@ARTICLE{10678877,
  author={Rostami Ghadi, Farshad and others},
  journal={IEEE Trans. Wireless Commun.}, 
  title={A Gaussian Copula Approach to the Performance Analysis of Fluid Antenna Systems}, 
month={Sep.},
  year={2024},
  volume={23},
  number={11},
  pages={17573-17585}}

@ARTICLE{10319727,
  author={Hou, Yanzhao and others},
  journal={IEEE Wireless Commun. Lett.}, 
  title={A Copula-Based Approach to Performance Analysis of Fluid Antenna System With Multiple Fixed Transmit Antennas}, 
month={Nov.},
  year={2024},
  volume={13},
  number={2},
  pages={501-504}}

@inproceedings{pakravan2024robust,
	title={Robust resource allocation for over-the-air computation networks with fluid antenna array},
	author={Pakravan, Saeid and others},
	booktitle={Proc. IEEE Globecom Workshops},
	month={},
	pages={},
	year={},
note={{C}ape {T}own, {S}outh {A}frica, pp. 1--6, Sep. 2024.}
}

@inproceedings{saeidpwcnc,
	title={Enhanced Over-the-Air Federated Learning Using {AI}-Based Fluid Antenna System},
	author={Ahmadzadeh, Mohsen and others},
	booktitle={Proc. IEEE WCNC},
	pages={},
	month={},
	year={},
	note={{M}ilan, Italy, pp. 1--6, May 2025.}
}

@article{ahmadzadeh2025ai,
	title={{AI}-based fluid antenna design for client selection in over-the-air federated learning},
	author={Ahmadzadeh, Mohsen and others},
	journal={IEEE Internet Things J.},
	volume={12},
	number={20},
	pages={42549--42558},
	year={2025},
	month= {Aug.},
	publisher={IEEE}
}

@article{zhang2024fluid,
	title={Fluid antenna array enhanced over-the-air computation},
	author={Zhang, Deyou and others},
	journal={IEEE Commun. Lett.},
	volume={13},
	number={6},
	pages={1541--1545},
	month={Mar.},
	year={2024},
	publisher={IEEE}
}

@ARTICLE{huangfu2025performance,
  author={Huangfu, Jiangsheng and others},
  journal={IEEE Trans. Wireless Commun.}, 
  title={Performance Analysis of Fluid Antenna System Under Spatially-Correlated Rician Fading Channels}, 
month={Jul.},
  year={2025},
  volume={25},
  number={},
  pages={1394-1407}}

@article{chen2023joint,
	title={Joint client selection and receive beamforming for over-the-air federated learning with energy harvesting},
	author={Chen, Caijuan and others},
	journal={IEEE Open J. Commun. Soc.},
	volume={7},
	volume={4},
	pages={1127--1140},
	year={2023},
	month={May},
	publisher={IEEE}
}

@book{nelsen2006introduction,
	title={An introduction to copulas},
	author={Nelsen, Roger B},
	year={2006},
	publisher={Springer}
}

@article{li2023integrated,
	title={Integrated sensing, communication, and computation over-the-air: {MIMO} beamforming design},
	author={Li, Xiaoyang and others},
	journal={IEEE Trans. Wireless Commun.},
	volume={22},
	number={8},
	pages={5383--5398},
	month={Jan.},
	year={2023},
	publisher={IEEE}
}

@ARTICLE{10906511,
  author={Zhu, Lipeng and others},
  journal={IEEE Commun. Surveys Tuts}, 
  title={A Tutorial on Movable Antennas for Wireless Networks}, 
  year={2026},
  month={Feb.},
  volume={28},
  number={},
  pages={3002-3054}}

@ARTICLE{10729877,
  author={Li, Nianzu and others},
  journal={IEEE Wireless Commun. Lett.}, 
  title={Over-the-Air Computation via {2-D} Movable Antenna Array}, 
  year={2025},
  month={Jan.},
  volume={14},
  number={1},
  pages={33-37}}

@article{silva2008copula,
  author    = {Silva, Ralph dos Santos and others},
  title     = {Copula, marginal distributions and model selection: A Bayesian note},
  journal   = {Statist. Comput.},
  volume    = {18},
  number    = {3},
  pages     = {313--320},
  year      = {2008},
  month     = {Mar.}
}

@ARTICLE{11106811,
  author={Pakravan, Saeid and others},
  journal={IEEE Trans. Veh. Technol.}, 
  title={Fluid Antenna-Assisted Uplink {NOMA} Networks Under Imperfect {SIC}}, 
  year={2026},
  month     = {Jan.},
  volume={75},
  number={1},
  pages={1689-1694}}

	\vfill
	
\end{document}